\documentclass[11pt,british]{article} 

\usepackage[T1]{fontenc}
\usepackage[latin9]{inputenc}
\usepackage[a4paper]{geometry}
\usepackage{units}
\usepackage{amsthm}
\usepackage{amsmath} 
\usepackage{graphicx}

\makeatletter
  \theoremstyle{plain}
  \newtheorem{prop}{\protect\propositionname}
 \ifx\proof\undefined\
   \newenvironment{proof}[1][\proofname]{\par
     \normalfont\topsep6\p@\@plus6\p@\relax
     \trivlist
     \itemindent\parindent
     \item[\hskip\labelsep
           \scshape
       #1]\ignorespaces
   }{%
     \endtrivlist\@endpefalse
   }
   \providecommand{\proofname}{Proof}
 \fi

\makeatother

\usepackage{babel}
\providecommand{\propositionname}{Proposition}

\begin{document}

\title{Simplicial Complex Entropy}

\author{
\parbox{0.45\textwidth}
{\centering
Stefan Dantchev\\
Durham University\\
School of Eng. \& Comp. Sciences\\
Durham, DH1 3LE, UK\\
s.s.dantchev@durham.ac.uk
}
\hspace{0.09\textwidth}
\parbox{0.45\textwidth}{\centering
Ioannis Ivrissimtzis\\
Durham University\\
School of Eng. \& Comp. Sciences\\
Durham, DH1 3LE, UK\\
ioannis.ivrissimtzis@durham.ac.uk
}
}

\date{}

\maketitle

\begin{abstract} 
We propose an entropy function for simplicial complices. Its value 
gives the expected cost of the optimal encoding of sequences of 
vertices of the complex, when any two vertices belonging to 
the same simplex are indistinguishable. We show that the proposed entropy 
function can be computed efficiently. By computing the entropy of 
several complices consisting of hundreds of simplices, we show that 
the proposed entropy function can be used in the analysis of the 
large sequences of simplicial complices that often appear in 
computational topology applications. 
\end{abstract}

\maketitle

\section{Introduction} 
\label{sec:introduction}

In several fields of visual computing, such as computer vision, 
CAD and graphics, many applications require the processing of an 
input in the form of a set of unorganized points, that is, a finite 
subset of a metric space, typically ${\bf R}^2$ or ${\bf R}^3$. Often, 
the first step in the processing pipeline is the construction of a 
simplicial complex, or a series of simplicial complices capturing 
spatial relations of the input points. Such geometrically constructed 
simplicial complices commonly used in practice include the 
{\em Vietoris-Rips} and {\em \v{C}ech} complices, see for example 
\cite{zomorodian10}, the {\em alpha shapes} \cite{edelsbrunner95} 
and the witness complices \cite{silva04,guibas07}. 

The two simplest constructions, giving the Vietoris-Rips and the 
\v{C}ech complices, emerged from studies in the field of algebraic 
topology. In the Vietoris-Rips construction, we connect two points 
with an edge if their distance is less than a fixed $\varepsilon$ 
and the simplices of the complex are the cliques the resulting graph. 
In the \v{C}ech construction, the simplices are the sets of vertices 
that lie inside a bounding sphere of radius $\varepsilon$. 

Notice that the complices constructed this way, apart from the input 
point set which gives their vertex set, also depend on the parameter 
$\varepsilon$. In applications where the goal is to extract topological 
information related to input point set, it is quite common to consider 
sequences of complices corresponding to different values of $\varepsilon$ 
and study the evolution of their topological properties as $\varepsilon$ 
varies \cite{Zomorodian05,johansson11}. Such investigations led to the 
development of the notion of {\em persistence}, in the form for example 
of persistent homology, as one of the main concepts in the field of 
computational topology \cite{Edelsbrunner00,silva07}. Indicative of the 
need for computational efficiency, persistent homology calculations 
based on millions of distinct complices from the same input point 
set are now common and thus, the efficient computations of such series 
of complices is an active research area \cite{zomorodian10,dantchev12}. 

In this paper, our aim it to use information theoretic tools to study 
sequences of geometrically constructed complices corresponding to 
different values of $\varepsilon$. In particular, we define an entropy 
function on simplicial complices; we show that it can be computed 
efficiently; and demonstrate that it can be used to find critical 
values of $\varepsilon$. Here, the value of $\varepsilon$ is seen 
as a measure of spatial resolution and thus, we interpret the simplices 
of the geometrically constructed complices as sets of indistinguishable 
points. 

The setting of our problem is very similar to one that gave rise to the 
concept of {\em graph entropy} \cite{korner73} and {\em hypergraph entropy} 
\cite{korner88}. There, a graph or a hypergraph describe indistinguishability 
relations between vertices and the sets of indistinguishable vertices are 
derived as the {\em independent sets} of the graph or hypergraph. In contrast, 
in our approach, the sets of indistinguishable vertices are readily given as 
the simplices of the complex. In the next section, immediately after 
introducing the proposed simplicial complex entropy, we discuss in more 
detail its relation to graph entropy.

% **********************************************************************
% **********************************************************************
% **********************************************************************

\section{Simplicial complex entropy} 
\label{sec:entropy}

Let $V=\left\{ v_{1},\dots v_{n}\right\}$ be a point set consisting of $n$ 
vertices. An abstract simplicial complex $C$ over $V$ is given by its maximal 
simplices $C_{1},\dots C_{m}$. These are nonempty subsets of $V$ whose union 
is the entire $V$ and none of them is a subset of another.

We are also given a probability distribution $P$ over $V$, i.e. non-negative 
numbers $p_{1},\dots p_{n}$ and such that $\sum_{j=1}^{n}p_{j}=1$. 
Assuming that all points that belong to the same simplex $C_{i}$ for 
some $i$, $1\leq i\leq m$ are indistinguishable, we define the 
{\em simplicial complex entropy} as 
\begin{eqnarray}
H(C,P) = \min\sum_{j=1}^{n}p_{j}\log\frac{1}{{\displaystyle \sum_{i\in Simpl\left(j\right)}}q_{i}} & \mbox{s.t.}\label{eq:Entropy_optimisation}\\
\sum_{i=1}^{m}q_{i}=1\nonumber \\
q_{i}\geq0 & \mbox{} & 1\leq i\leq m.\nonumber 
\label{eq:entropy}
\end{eqnarray}
where $Simpl\left(j\right)$ denotes the set of simplices containing vertex $p_j$. 

The above simplicial complex entropy is similar to the graph entropy, defined 
over a graph $G$ with a probability distribution $P$ on its vertices, given by 
\begin{equation}
\min\sum_{j=1}^{n}p_{j}\log\frac{1}{a_i} 
\label{eq:graphEntropy}
\end{equation}
where the minimum is taken over all convex combinations of characteristics 
vectors {\bf a} of the independent sets of $G$, with $a_i$ denoting the 
$i$-th coordinate of such vectors. 

In its information theoretic interpretation, the graph entropy gives the 
expected number of bits per symbol required in an optimal encoding of the 
information coming from a source emitting vertices of $G$ under the probability 
distribution $P$, assuming that any two vertices are indistinguishable iff 
they are not connected with an edge \cite{simonyi95}. In other words, the 
independent sets of $G$ are the sets of mutually indistinguishable vertices. 
Similarly, the information theoretic interpretation of the proposed simplicial 
complex entropy is that of the expected bits per vertex required in an 
optimal encoding of the information coming from the same source, under the 
assumption that the sets of mutually indistinguishable vertices are exactly 
the simplices. 

The proposed simplicial complex entropy can be seen as a simplification of 
the graph entropy, which however is at least as general. Indeed, on a graph 
$G$ we can define a simplicial complex $C$ on the same vertex set as $G$ and 
its simplices being the independent sets of $G$. Then, the graph entropy of 
$G$ is the simplicial complex entropy of $C$. On the other hand, given a 
simplicial complex $C$ it is not immediately obvious how one can construct a 
graph $G$ such that the simplicial complex entropy of $C$ is the graph entropy 
of $G$. 

In an abstract context, the proposed simplification might seem quite 
arbitrary: instead of deriving the sets of indistinguishable vertices 
from the connectivity of a graph, we consider them given in the form 
of simplices. However, in the context of geometrically constructed 
simplicial complices embedded in a metric space, the simplices are the 
natural choice of sets of indistinguishable points for a given spatial 
resolution $\varepsilon$ and there is no need, or indeed an obvious 
way, to model the property of indistinguishability in terms of graph 
connectivity. One notable exception to this is the special case of 
Vietoris-Rips complices which we discuss next, aiming at further 
highlighting differences and similarities between simplicial entropy 
and graph entropy.

% **********************************************************************
% **********************************************************************
% **********************************************************************

\subsection{Example: Vietoris-Rips simplicial complex entropy}

In the case of Vietoris Rips complices, there is a straightforward 
interpretation of the simplicial complex entropy as graph entropy. 
Indeed, assume a probability distribution $P$ on a set of vertices 
$V$ embedded in a metric space, and assume that two vertices,  
are indistiguishable if their distance is less than $\varepsilon$. 
The graph $G$ with its edges connecting pairs of distinguishable 
vertices is the complement of the underlying graph of the 
Vietoris-Rips complex constructed on $V$ for the same $\varepsilon$. 

It is easy to see that the independent sets of $G$ are exactly the 
simplices of the Vietoris-Rips complex and thus, the graph entropy 
of $G$ is the simplicial complex entropy of the Vietoris-Rips complex. 
Indeed, if there are no edges connecting points of a subset of $V$, 
it means that all distances between these points are less than 
$\varepsilon$, therefore they form a simplex of the Vietoris-Rips 
complex. 

% *************************************************************************

The simplicial entropy of the Vietoris-Rips complices has a straightforward 
graph entropy interpretation because Vietoris-Rips complices are completely 
defined by their underlying graph. Indeed, their simplices are the cliques 
of the underlying graph. However, this is not generally the case for 
geometrically constructed complices, with the \v{C}ech complex being a 
notable counterexample. 

Indeed, consider as $V$ the three vertices of an equilateral triangle 
of edgelength 1, embedded in ${\bf R}^2$. Any pair of vertices 
corresponds to an edge of the triangle and has a minimum enclosing 
sphere of radius 1/2. The $V$ itself has a minimum enclosing sphere of 
radius $\sqrt{3}/3$. Thus, for any $1/2 \leq \varepsilon \leq \sqrt{3}/3$ 
all three edges of the triangle are simplices of the \v{C}ech complex, i.e. 
pair-wise indistinguishable, but the triangle itself is not a simplex of 
the \v{C}ech complex.

\section{Properties of simplicial complex entropy} 
\label{sec:properties}

Solving the entropy minimisation turns out to be computationally tractable. Let us denote
\[
S_{j}\left(q\right)\overset{\mbox{def}}{=}{\displaystyle \sum_{i\in Simpl\left(j\right)}}q_{i}
\]
and rewrite (\ref{eq:Entropy_optimisation}) as a maximisation problem
with an objective function
\begin{equation}
f\left(q\right)\overset{\mbox{def}}{=}\sum_{j=1}^{n}p_{j}\log S_{j}\left(q\right).\label{eq:Objective_maximum}
\end{equation}
We can immediately prove the following
\begin{prop}
\label{prop:convexity}The objective function (\ref{eq:Objective_maximum})
is concave. The sums $S_{j}\left(q\right)$ are unique (i.e. the same)
for all points $q$ where the maximum is attained, while the set of
all maxima is a polyhedron.\end{prop}
\begin{proof}
Let $q'$ and $q''$ be two different feasible points. Clearly, the
point $q\overset{\mbox{def}}{=}\nicefrac{1}{2}\left(q'+q''\right)$
is also feasible and $S_{j}\left(q\right)=\nicefrac{1}{2}\left(S_{j}\left(q'\right)+S_{j}\left(q''\right)\right)$
for $1\leq j\leq n$. We then have
\begin{equation}
\log S_{j}\left(q\right)\geq\nicefrac{1}{2}\left(\log S_{j}\left(q'\right)+\log S_{j}\left(q''\right)\right),\label{eq:Log_is_concave}
\end{equation}
which proves the concavity of the objective function.

Imagine now that $q'$ and $q''$ are two (different) optimal points
(with $f\left(q'\right)=f\left(q''\right)$) and moreover there is
a $j$, $1\leq j\leq n$ such that $S_{j}\left(q'\right)\neq S_{j}\left(q''\right)$.
For that particular, (\ref{eq:Log_is_concave}) is a strict inequality
and after summing up all inequalities, we get
\[
f\left(q\right)>\nicefrac{1}{2}\left(f\left(q'\right)+f\left(q''\right)\right),
\]
which contradicts the optimality of both $q'$ and $q''$. Thus, the
sums $\log S_{j}\left(q\right)$ are unique over all optimal points
$q$.

Finally, if we denote these sums (at an optimum) by $s_{j}$, $1\leq j\leq n$,
we notice that the set of all optimal optimal points $q$ is fully
described by the following linear system:
\begin{eqnarray*}
{\displaystyle \sum_{i\in Simpl\left(j\right)}}q_{i}=s_{j} &  & 1\leq j\leq n\\
\sum_{i=1}^{m}q_{i}=1\\
q_{i}\geq0 & \mbox{} & 1\leq i\leq m.
\end{eqnarray*}

\end{proof}

Another useful characterisation of an optimal point $q$ is given
by
\begin{prop}
\label{prop:optimality-conditions}Any optimal point $q$ satisfies
the following ``polynomial complementarity'' system:
\begin{eqnarray*}
\sum_{j\in Pts\left(i\right)}\frac{p_{j}}{S_{j}\left(q\right)}\begin{cases}
=1 & \mbox{if }q_{i}>0\\
\leq1 & \mbox{if }q_{i}=0
\end{cases} &  & 1\leq i\leq m\\
\sum_{i=1}^{m}q_{i}=1\\
q_{i}\geq0 & \mbox{} & 1\leq i\leq m
\end{eqnarray*}
\end{prop}
\begin{proof}
The gradient of the objective function, $\nabla f\left(q\right)$
is
\[
\left(\sum_{j\in Pts\left(1\right)}\nicefrac{p_{j}}{S_{j}\left(q\right)},\dots\sum_{j\in Pts\left(m\right)}\nicefrac{p_{j}}{S_{j}\left(q\right)}\right)^{T}.
\]
We start with Karush\textendash{}Kuhn\textendash{}Tucker conditions
(for the maximisation problem) that an optimal point $q$ should satisfy:

\begin{eqnarray}
\sum_{j\in Pts\left(i\right)}\frac{p_{j}}{S_{j}\left(q\right)}=\lambda-\mu_{i} &  & 1\leq i\leq m\label{eq:Partial_derivatives}\\
\sum_{i=1}^{m}q_{i}=1\label{eq:Equal_to_one}\\
q_{i},\mu_{i}\geq0\; q_{i}\mu_{i}=0 & \mbox{} & 1\leq i\leq m
\label{eq:Non_negative}
\end{eqnarray}
for some $\lambda$ and $\mu_{i}$, $1\leq i\leq m$. 

We first expand the inner product
\begin{eqnarray}
\left\langle q,\nabla f\left(q\right)\right\rangle = 
\sum_{i=1}^{m}q_{i}\sum_{j\in Pts\left(i\right)}\frac{p_{j}}{S_{j}\left(q\right)} = & & \\ 
= \sum_{j=1}^{n}\frac{p_{j}}{S_{j}\left(q\right)}\sum_{i\in Simpl\left(j\right)}q_{i}=\sum_{j=1}^{n}p_{j}=1. & & 
\end{eqnarray} 
On the other hand, from \ref{eq:Partial_derivatives}, \ref{eq:Equal_to_one}
and \ref{eq:Non_negative}, we get
\begin{eqnarray}
\sum_{i=1}^{m}q_{i}\sum_{j\in Pts\left(i\right)}\frac{p_{j}}{S_{j}\left(q\right)}
=\sum_{i=1}^{m}q_{i}\left(\lambda-\mu_{i}\right)= & & \\
= \lambda\sum_{i=1}^{m}q_{i}-\sum_{i=1}^{m}q_{i}\mu_{i}=\lambda, & & 
\end{eqnarray} 
and thus $\lambda=1$.
\end{proof}

% ****************************************************************
% ****************************************************************
% ****************************************************************

\subsection{Defining the error} 
\label{sec:error}

The indistiguashability between points, as described by the simplicial complex, 
results into an error of a complex can be understood in terms of encoding and 
decoding points as follows.

The encoder gets point $j$ ($1\leq j\leq n$), which is produced
by a memoryless random source under distribution $p$. We will describe 
two encoding startegies, one randomised, which is the one we implemented, 
and an adversarial which should give higher error rates.

\subsubsection{Randomised encoder}

The randomised encoder produces one of the cells that contains $j$, under 
the distribution $\nicefrac{q_{i}}{S_{j}\left(q\right)}$ for all $i\in Simpl\left(j\right)$.
The overall probability of seeing cell $i$ as a result is
\begin{equation}
\sum_{j\in Pts\left(i\right)}p_{j}\frac{q_{i}}{S_{j}\left(q\right)}=q_{i}\sum_{j\in Pts\left(i\right)}\frac{p_{j}}{S_{j}\left(q\right)}=q_{i}\left(1-\mu_{i}\right)=q_{i}
\label{eq:errorRandomised}
\end{equation}
(where $\mu_{i}$ is as in the proof of proposition \ref{prop:optimality-conditions}
above and taking into account that $\lambda=1$) as expected.

The decoder sees a cell $i$ and its best guess (as to which point
actually produced it) is the one that has the biggest probability
(according to the distribution $p$). Thus the total gain is
\[
err = \sum_{i=1}^{m}q_{i}\frac{\max_{j\in Pts\left(i\right)}p_{j}}{\sum_{j\in Pts\left(i\right)}p_{j}}.
\]

\subsubsection{The adversarial encoder}

We can think of this encoding strategy as a game between the encoder and the
decoder, in which whenever the decoder sees a simplex $i$, he responds
with a guess (point) $j\in Simpl\left(i\right)$ according to probabilities
$r_{ij}$, $r_{ij}\geq0$ and such that
\[
\sum_{j\in Pts\left(i\right)}r_{ij}=1\;\mbox{ for every }1\leq i\leq m.
\]
These probabilities are known to the encoder, so if the source produced
a point $j$, the encoder minimises the gain of the decoder by picking
a cell $i$ that is $\arg\min_{j\in Simpl\left(i\right)}r_{ij}$.
In turn, the decoder tries to maximise their total expected gain as
\begin{eqnarray*}
\overline{err} =  \max\sum_{j=1}^{n}p_{j}r_{j} & \mbox{s.t.}\\
r_{ij}\geq r_{j} &  & 1\leq j\leq n\mbox{ and }i\in Simpl\left(j\right)\\
\sum_{j\in Pts\left(i\right)}r_{ij}=1 &  & 1\leq i\leq m\\
r_{ij}\geq0 &  & 1\leq j\leq n\mbox{ and }i\in Simpl\left(j\right)
\end{eqnarray*} 
\section{Examples} 
\label{sec:examples} 

The computation of the simplicial complex entropy and the error was implemented 
in Matlab. Apart from some code for input output operations and simplicial 
complex representation, {\em fmincon} and {\em linprog} were directly used 
to compute the entropy and the error, respectively. 

In all examples, we report the {\em normalised entropy}, that is, the 
simplicial complex entropy $H(C,P)$ divided by the entropy of the vertex set 
$V$ under the same probability distribution $P$. Instead of the randomised 
encoder error $err$ in Eq.~\ref{eq:errorRandomised}, we report the value 
$1-err$, which can be seen as the decoding accuracy rate and correlates 
nicely with the normalised entropy. The difference between these two values 
is also reported. 

\begin{figure}[t]
\centering
\includegraphics[width=0.24\textwidth]{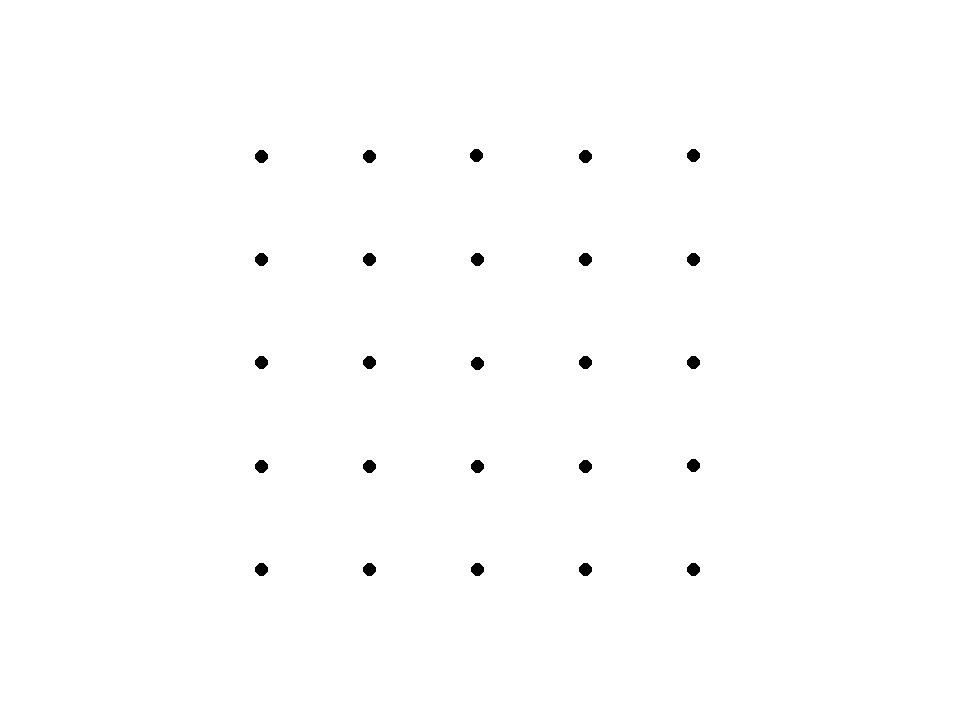} \hfill
\includegraphics[width=0.24\textwidth]{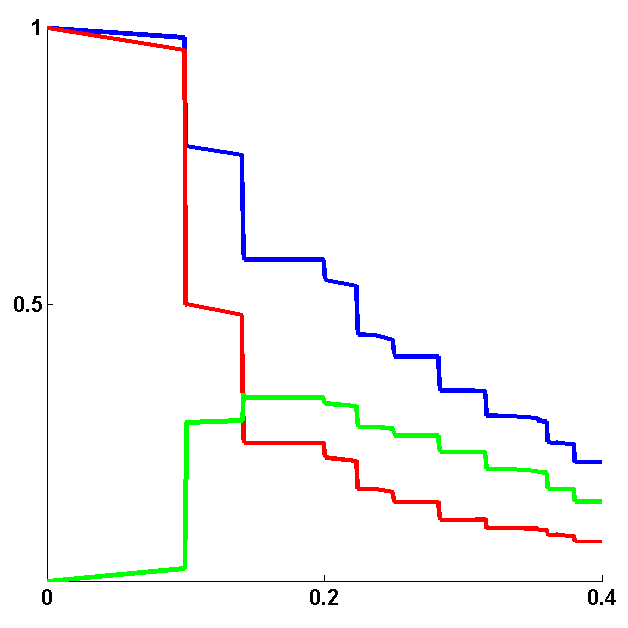} \hfill 
\includegraphics[width=0.24\textwidth]{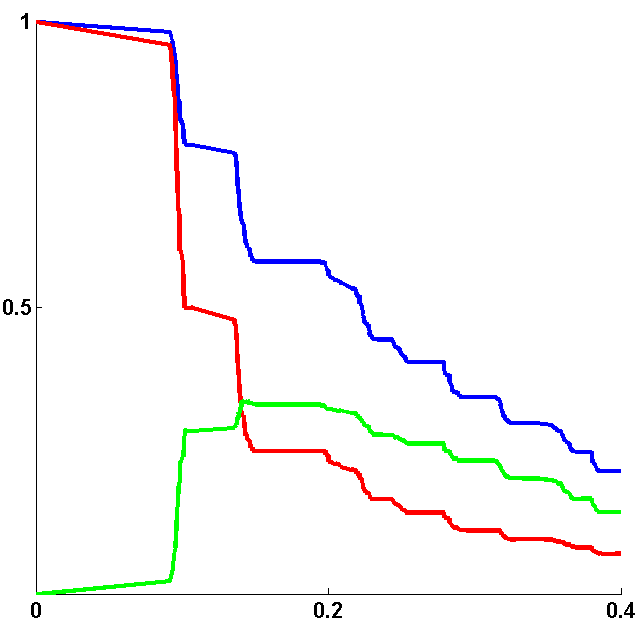} \hfill
\includegraphics[width=0.24\textwidth]{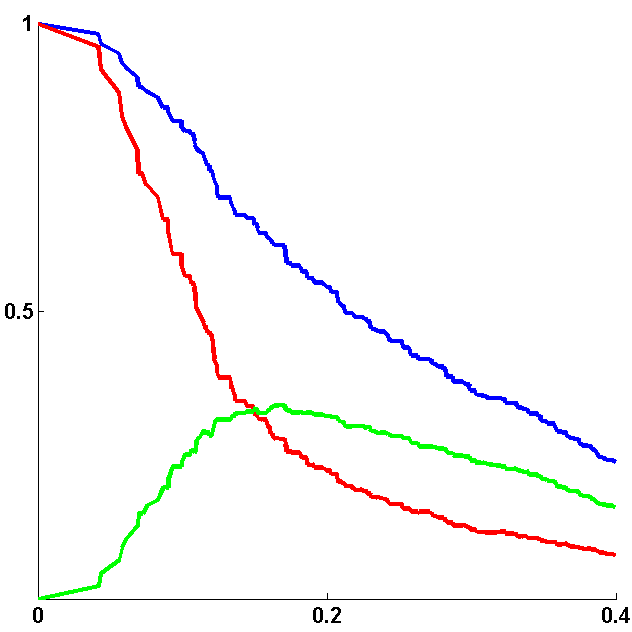} \vskip 0.1cm
\includegraphics[width=0.24\textwidth]{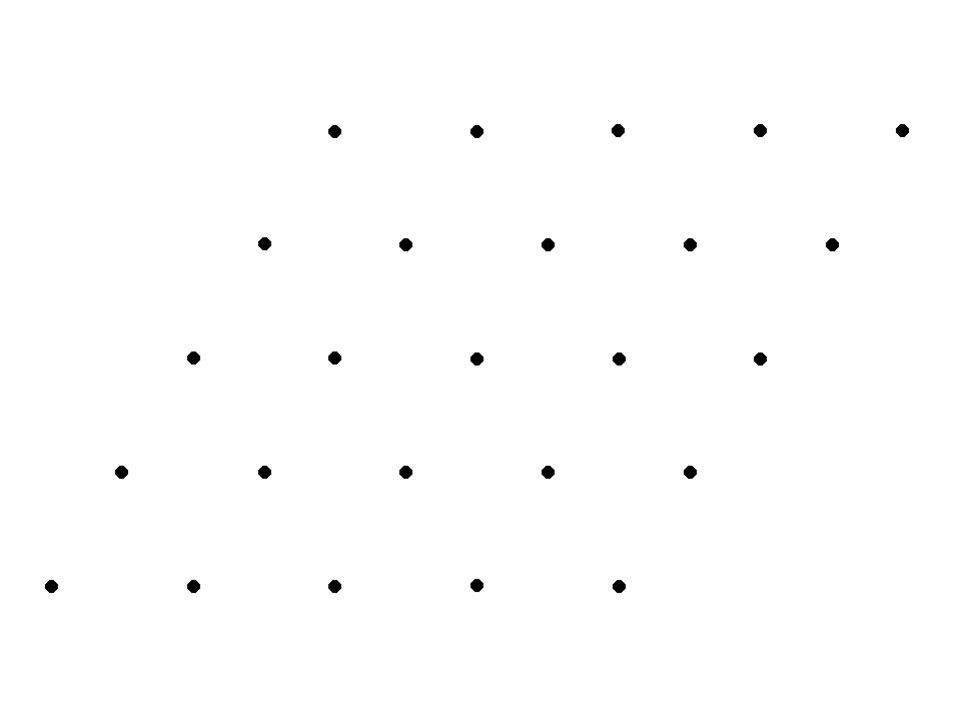} \hfill
\includegraphics[width=0.24\textwidth]{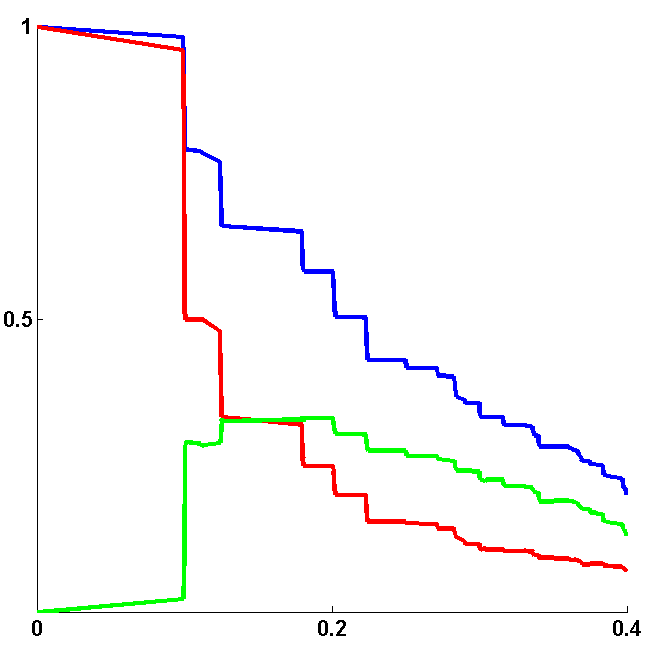} \hfill 
\includegraphics[width=0.24\textwidth]{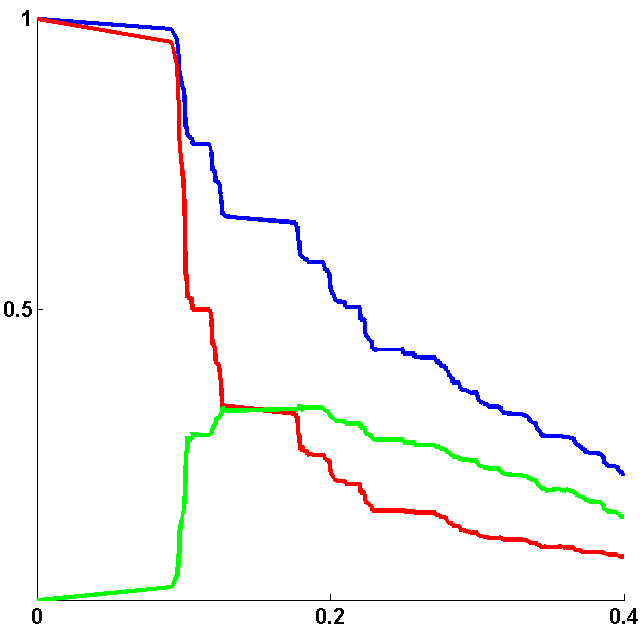} \hfill
\includegraphics[width=0.24\textwidth]{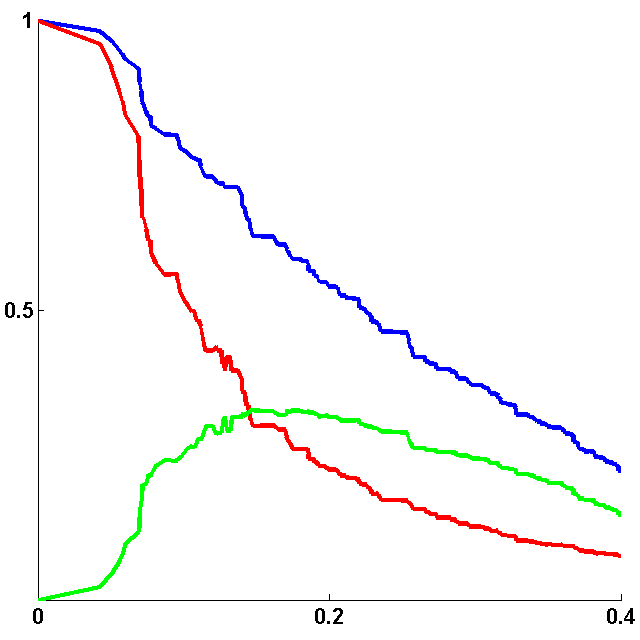} 
\caption{The $y$-axis represents the normalised entropy (blue curve), the accuracy rate (red curve) and their difference (green curve). The $x$-axis represents the parameter $\varepsilon$ (radius of the minimal enclosing sphere of a simplex) in the construction of the \v{C}ech complex. {\bf Top:} The input point set is the $5\times 5$ block of vertices of a square grid of edgelength 0.2 shown in the left. From left to right, uniform random noise $\pm 0.5\%, \pm 5\%$ and $\pm 50\%$ of the edgelength was added. The figures represent entropy and error computations on all possible \v{C}ech complices for that range of $\varepsilon$, that is, 768, 746 and 685 distinct complices, respectively. {\bf Bottom:} As per the top, but for triangular grid points. The figures correspond to 725, 694 and 672 distinct complices, respectively.}
\label{fig:latticeData}
\end{figure}

In a first example, Figure~\ref{fig:latticeData} (Top) shows the values of 
these two functions on \v{C}ech complices constructed from vertex sets that 
are nodes of square grid of edgelength 0.2 with some added noise. 
Figure~\ref{fig:latticeData} (Bottom) shows a similar example with the 
vertices originally being nodes of a triangular grid. In all cases, the 
probability distribution $P$ on the vertex set is uniform. 

In the case of a square grid without any added noise, as the values of the 
parameter $\varepsilon$ of the \v{C}ech complex construction parameter 
increase, they reach the first critical value at $\varepsilon = 0.1$, when 
edges, i.e. simplices of degree 2, are formed. The next critical value is 
$\varepsilon \simeq 0.141$, where the simplices of 
degree 4 are formed, and the next critical value is $\varepsilon = 0.2$ when 
simplices of degree 5 are formed. Similarly, the first critical values 
in the case of points from a triangular grid are $\varepsilon=0.1$, when 
simplices of degree 2 are formed and $\varepsilon \simeq 0.115$ 
when simplices of degree 3 are formed. 

These critical values are shown Figure~\ref{fig:latticeData} as sudden drops 
in the entropy of the \v{C}ech complices constructed on the less noisy data 
sets. We also notice simultaneous drops of the accuracy rates since they, as 
expected, correlate well with entropy. As the level of noise increases the 
critical points become less visible on either of these two curves. However, 
their difference, shown in green, seems to be more robust against noise, and 
moreover, seems to peak at a favorable place. That is, it peaks in values of 
$\varepsilon$ that would neither return a large number of non-connected 
components nor heavily overlapping simplices.

% ***********************************************************************************

In the second example, the input set is a sample from the unit sphere in 
${\bf R}^3$. Figure~\ref{fig:sphericalData} (left) shows results from 
regular samples of size 50 (top) and 100 (bottom), computed in \cite{wales06} 
as solutions to the Thomson problem, with added uniform noise of $\pm 0.01$ units. 
In \cite{wales06}, the minimum distances between a point and its nearest neighbour 
are $\sim 0.5$ and $\sim $, respectively, and correspond to the steep entropy 
decreases and the half of these values when the first edges of the \v{C}ech complices 
are formed. Figure~\ref{fig:sphericalData} (right) shows results from random, 
area uniform samples of size 50 (top) and 100 (bottom). While the input is much less 
regular than at the left hand side of the figure, the peaks of the two green curves 
align well. 
\begin{figure}[t]
\centering
\includegraphics[width=0.24\columnwidth]{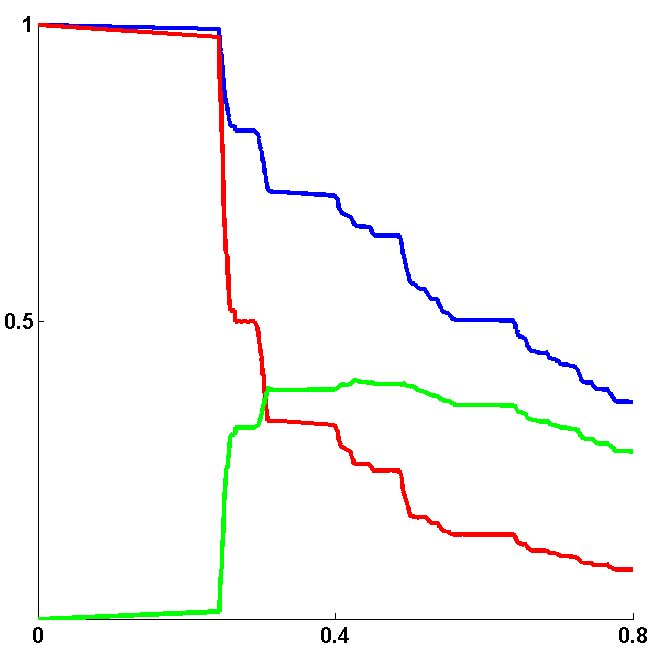} \hfill
\includegraphics[width=0.24\columnwidth]{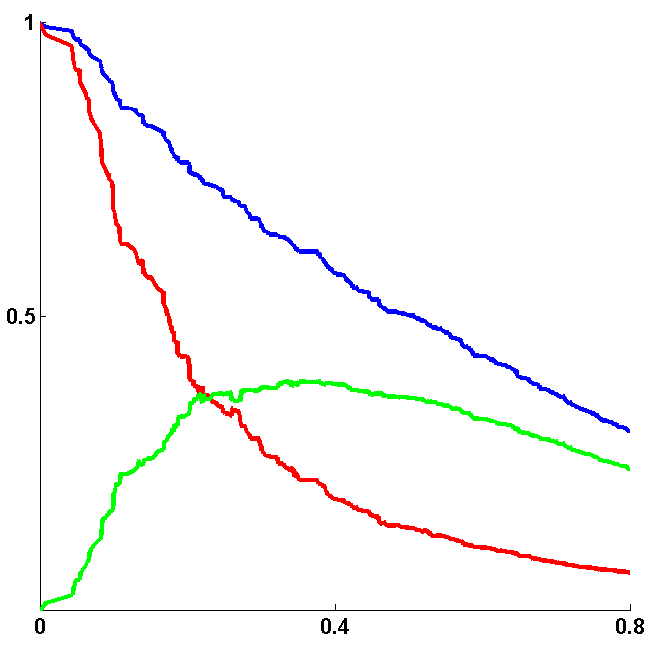} \hfill
\includegraphics[width=0.24\columnwidth]{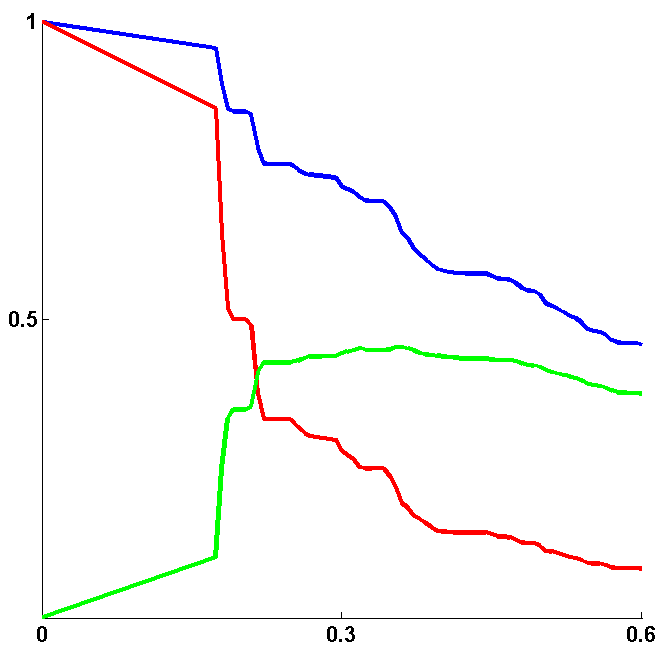} \hfill
\includegraphics[width=0.24\columnwidth]{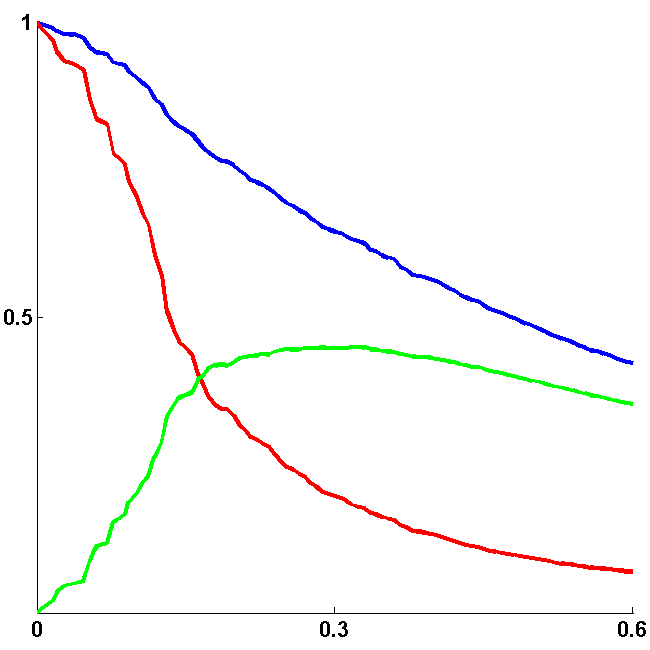} 
\caption{The axes and the colour of the curves are as per Figure~\ref{fig:latticeData}. {\bf Two left figures:} The input point set of size 50 is a computational solution to the Thomson problem with uniform noise of $\pm 0.01$ units added on each coordinate. In the right figure the input in an area uniform spherical random sample of the same size. The figures represent entropy and error computations from 2523 and 2661 distinct \v{C}ech complices, respectively. {\bf Two right figures:} As per the top, with point sets of size 100. Due to the very large number of distinct \v{C}ech complices, each figure represents 100 \v{C}ech complices, corresponding to a uniform sample of values of $\varepsilon$ in [0,0.6].} 
\label{fig:sphericalData}
\end{figure}

% ***********************************************************************************

In a third example, we solve the optimisation problem for the 
computation of the entropy on triangle meshes and show the values 
of $q$, as in Eq.~\ref{eq:entropy}, color-mapped on the mesh triangles. 
In Figure~\ref{fig:triangleMeshes} (left), the probability distribution 
on the mesh vertices is uniform, as it was in all previous examples. 
On the right hand side of the figure, the probability distribution follows 
the absolute value of the discrete Gaussian curvature of the vertices. 
\begin{figure}[t]
\centering
\includegraphics[width=0.24\textwidth]{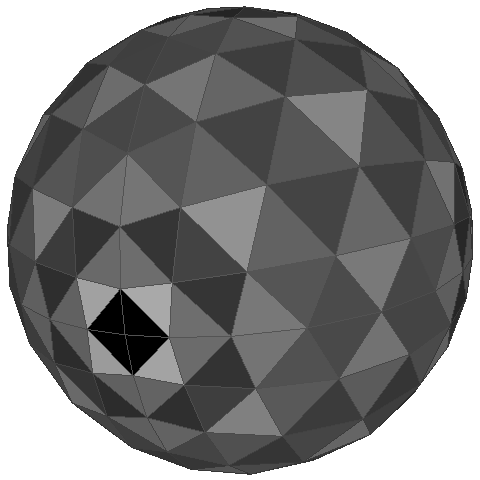} \hfill
\includegraphics[width=0.24\textwidth]{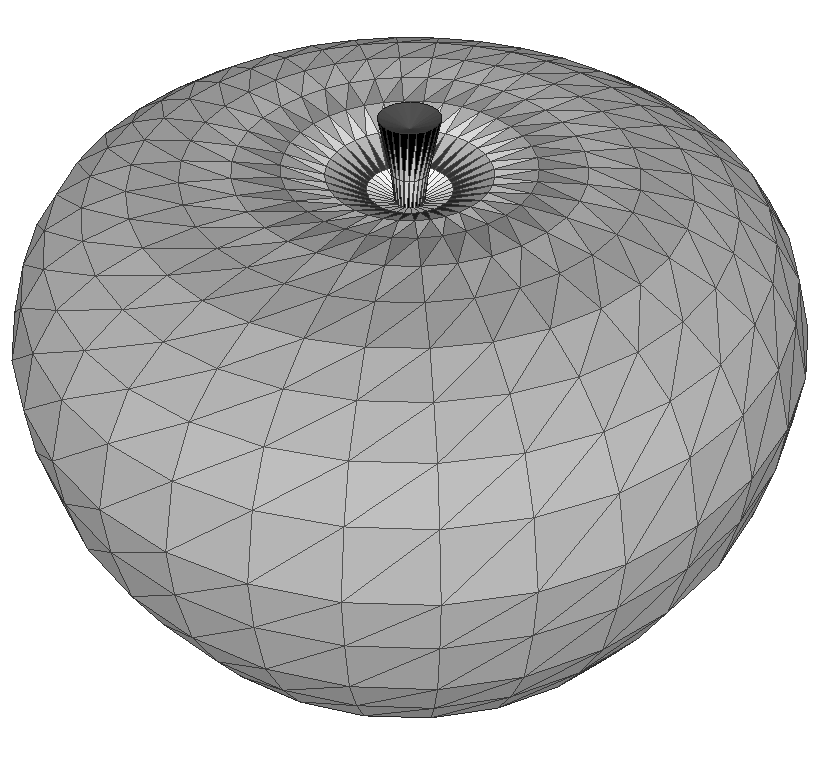} \hfill 
\includegraphics[width=0.24\textwidth]{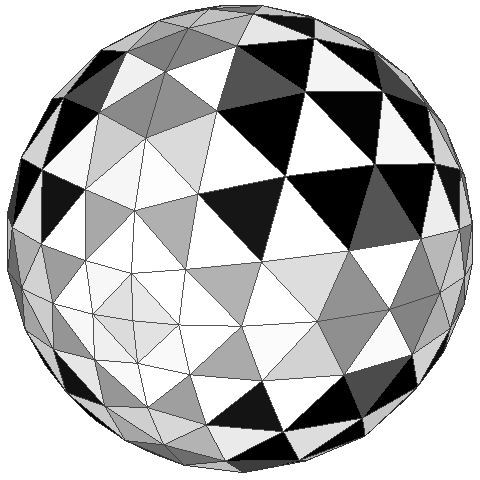} \hfill
\includegraphics[width=0.24\textwidth]{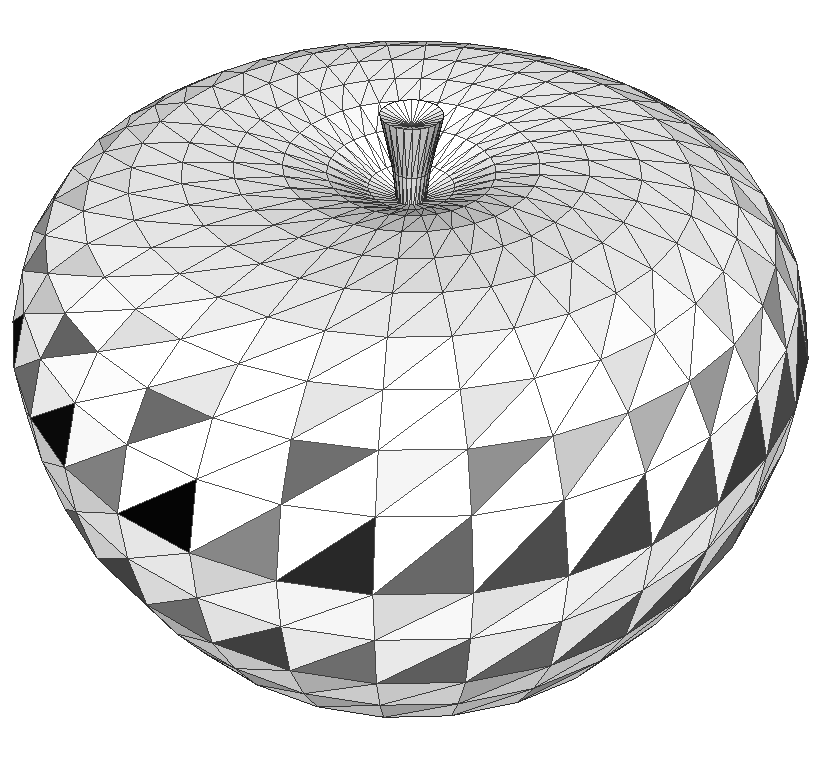} 
\caption{{\bf Two left figures:} The values $q_i$ in Eq.~1 are color-mapped on the 
mesh triangles. Darker colors correspond to higher values. {\bf Two right figures:} The probability 
distribution $P$ on the vertices corresponds to the absolute values of the discrete 
Gaussian curvature of the vertices. The two meshes consist of 512 and 1704 triangle, respectively.} 
\label{fig:triangleMeshes}
\end{figure} 
\section{Conclusion} 

We presented an entropy function for simplicial complices which can be seen 
as a simplification and generalisation of the graph entropy since all the 
maximal sets of indistiguishable points are exactly the maximal simplices of 
the complex and do not have to be computed as the independent sets of the graphs, 
which, generally, are difficult to characterise. We show that this simplification 
makes the simplicial complex entropy an function that can be efficiently computed. 

Even though the entropy is defined on abstract simplicial complexes, which are 
purely topological structures, in the examples we show that it can be relevant 
to geometric applications. For example, by computing the entropy of geometrically 
constructed simplicial complices, such as the {\em \v{C}ech} complices, or by 
using geometric properties of an embedded complex, such as a discrete curvature 
computed on the vertices to obtain a probability distribution on them. 

In the future we would like to study in more detail the function given as the 
difference between normalised entropy and the decoding accuracy rates, which 
seems to be a robust to noise descriptor of an appropriate level of geometric 
detail defined by the variable $\varepsilon$ of the {\em \v{C}ech} complex. We 
would also like to study the relationship between the error corresponding to 
a randomised encoder we used here and the error corresponding to an adversarial 
encoder as discussed at the end of Section~\ref{sec:error}.  

\bibliographystyle{IEEEtran}
\bibliography{ref}

\begin{thebibliography}{10}
\providecommand{\url}[1]{#1}
\csname url@samestyle\endcsname
\providecommand{\newblock}{\relax}
\providecommand{\bibinfo}[2]{#2}
\providecommand{\BIBentrySTDinterwordspacing}{\spaceskip=0pt\relax}
\providecommand{\BIBentryALTinterwordstretchfactor}{4}
\providecommand{\BIBentryALTinterwordspacing}{\spaceskip=\fontdimen2\font plus
\BIBentryALTinterwordstretchfactor\fontdimen3\font minus
  \fontdimen4\font\relax}
\providecommand{\BIBforeignlanguage}[2]{{%
\expandafter\ifx\csname l@#1\endcsname\relax
\typeout{** WARNING: IEEEtran.bst: No hyphenation pattern has been}%
\typeout{** loaded for the language `#1'. Using the pattern for}%
\typeout{** the default language instead.}%
\else
\language=\csname l@#1\endcsname
\fi
#2}}
\providecommand{\BIBdecl}{\relax}
\BIBdecl

\bibitem{zomorodian10}
A.~Zomorodian, ``Fast construction of the vietoris-rips complex,''
  \emph{Computers \& Graphics}, 2010.

\bibitem{edelsbrunner95}
H.~Edelsbrunner, ``{The union of balls and its dual shape},'' \emph{Discrete
  and Computational Geometry}, vol.~13, no.~1, pp. 415--440, 1995.

\bibitem{silva04}
V.~de~Silva and G.~Carlsson, ``{Topological estimation using witness
  complexes},'' in \emph{Eurographics Symposium on Point-Based Graphics},
  M.~Alexa and S.~Rusinkiewicz, Eds., 2004.

\bibitem{guibas07}
L.~J. Guibas and S.~Y. Oudot, ``Reconstruction using witness complexes,'' in
  \emph{Proceedings of the eighteenth annual ACM-SIAM symposium on Discrete
  algorithms}, ser. SODA '07.\hskip 1em plus 0.5em minus 0.4em\relax
  Philadelphia, PA, USA: SIAM, 2007, pp. 1076--1085.

\bibitem{Zomorodian05}
A.~Zomorodian and G.~Carlsson, ``Computing persistent homology,''
  \emph{Discrete Comput. Geom.}, vol.~33, no.~2, pp. 249--274, 2005.

\bibitem{johansson11}
M.~Vejdemo-Johansson, ``Interleaved computation for persistent homology,''
  \emph{CoRR}, vol. abs/1105.6305, 2011.

\bibitem{Edelsbrunner00}
H.~Edelsbrunner, D.~Letscher, and A.~Zomorodian, ``Topological persistence and
  simplification,'' in \emph{FOCS '00}.\hskip 1em plus 0.5em minus 0.4em\relax
  IEEE, 2000, p. 454.

\bibitem{silva07}
V.~de~Silva and R.~Ghrist, ``{Coverage in sensor networks via persistent
  homology},'' \emph{Algebraic and Geometric Topology}, vol.~7, pp. 339--358,
  2007.

\bibitem{dantchev12}
S.~Dantchev and I.~Ivrissimtzis, ``Efficient construction of the {\v{c}}ech
  complex,'' \emph{Computers \& Graphics}, vol.~36, no.~6, pp. 708--713, 2012.

\bibitem{korner73}
J.~K{\"o}rner, ``Coding of an information source having ambiguous alphabet and
  the entropy of graphs,'' in \emph{6th Prague conference on information
  theory}, 1973, pp. 411--425.

\bibitem{korner88}
J.~Korner and K.~Marton, ``New bounds for perfect hashing via information
  theory,'' \emph{European Journal of Combinatorics}, vol.~9, no.~6, pp.
  523--530, 1988.

\bibitem{simonyi95}
G.~Simonyi, ``Graph entropy: a survey,'' \emph{Combinatorial Optimization},
  vol.~20, pp. 399--441, 1995.

\bibitem{wales06}
D.~J. Wales and S.~Ulker, ``Structure and dynamics of spherical crystals
  characterized for the thomson problem,'' \emph{Physical Review B}, vol.~74,
  no.~21, p. 212101, 2006.

\end{thebibliography}

\end{document}